\documentclass[12pt]{amsart}
\usepackage[utf8]{inputenc}
\usepackage[T1]{fontenc}
\DeclareMathAlphabet{\mathpzc}{OT1}{pzc}{m}{it}
\usepackage{geometry}
\usepackage{comment}
\usepackage{hyperref}

\usepackage{subcaption}
\usepackage{graphicx}
\usepackage[export]{adjustbox}
\usepackage{amsfonts}
\usepackage{amssymb}
\usepackage{amsthm}
\usepackage{amsmath}
\usepackage{amscd}
\usepackage[shortlabels]{enumitem}
\usepackage{mathrsfs}
\usepackage{tikz}
\usetikzlibrary{calc,arrows,decorations.pathreplacing}
\usepackage{nicefrac, xfrac}
\usepackage{mathtools,xparse}

\theoremstyle{plain}

\newtheorem{theorem}{Theorem}[section]

\newtheorem{proposition}[theorem]{Proposition}

\theoremstyle{definition}
\newtheorem{definition}[theorem]{Definition}

\newtheorem{remark}[theorem]{Remark}

\newtheorem*{theorem*}{Theorem}

\renewcommand{\phi}{\varphi}
\renewcommand{\epsilon}{\varepsilon}

\newcommand{\vertiii}[1]{{\left\vert\kern-0.25ex\left\vert\kern-0.25ex\left\vert #1
		\right\vert\kern-0.25ex\right\vert\kern-0.25ex\right\vert}}

\DeclareSymbolFont{bbold}{U}{bbold}{m}{n}
\DeclareSymbolFontAlphabet{\mathbbold}{bbold}

\newcommand{\NN}{\ensuremath{\mathbb N}}

\newcommand{\RR}{\ensuremath{\mathbb R}}

\def\N{\mathbb{N}}

\date{\today}

\title{Dynamical properties of chimera states for globally coupled map lattices}

\author{Th\'eophile Caby}
\address{CMUP, Departamento de Matem\`atica, Faculdade de C\^iencias, Universidade do Porto,
Rua do Campo Alegre s/n, 4169007 Porto, Portugal.}
\email{caby.theo@gmail.com} 
\author{Pierre Guiraud}
\address{Instituto de Ingeniería Matemática, CIMFAV, Universidad de Valparaiso, Valparaíso, Chile.}
\email{pierre.guiraud@uv.cl}

\begin{document}
\date{}
\maketitle
\begin{abstract}
We study the stability properties and long-term dynamical behavior of chimera states in globally coupled map lattices. In particular, we give a formula for the transverse Lyapunov exponent associated with blocks of synchronized sites. We use these results to study clustered dynamics from a numerical perspective, and give numerical evidence of attracting chimeras having chaotic dynamics, as well as periodic behaviors. Finally, we obtain some results ruling out the existence of absolutely continuous invariant measures supported on chimera states in strong coupling  regimes.
\end{abstract}

\section{Introduction}
Different definitions for chimera states are proposed in the literature on coupled map lattices (CML). They are generally described as the emergence of a block of synchronized sites that persists with the dynamics, while the rest of the sites are not synchronized with this block. Such dynamical patterns have been experimentally observed in real-world systems of similar interacting units, such as networks of chemical and mechanical oscillators \cite{exp1,exp2}. Some theoretical chimeras have been constructed in different systems of non-locally coupled oscillators \cite{stro,bick}, and are usually associated with a breaking of symmetry caused by the non-local nature of the coupling. For CMLs with important symmetries in the coupling structure, such as the one considered in this paper, chimeras were, until recently, widely believed to be impossible to arise from a generic initial configuration \cite{stro}. However, in a recent publication \cite{cosenza}, numerical evidence of the emergence of attracting chimera states for a globally coupled map lattice was presented.

The goal of this paper is to provide analytical and numerical tools to study  the eventual formation and persistence of chimeras states in globally coupled map lattices.
We start by introducing the system and proposing a  definition of cluster and chimera states.
We are then interesting in the stability  of these states. The important symmetries of the system allow us to decompose the tangent space in a convenient way, which renders an analytical treatment possible. In particular, it allows computing explicitly the transverse Lyapunov exponent associated with synchronizing blocks of arbitrary sizes. We then take advantage of these results to study them from a numerical perspective. This will lead us to present several examples of chimeras attracting a large set of initial conditions. Finally, in the last section, we study the long-term behavior of chimeric dynamics and exhibit regimes for which chimeric dynamics cannot preserve an absolutely continuous invariant measure. Overall, our results suggest that, for such globally coupled CMLs, attracting chimera states tend to evolve towards sets of zero Lebesgue measures, such as periodic cycles or hyperbolic attractors.

\section{Presentation of the problem}
We consider a coupled map lattice of finite size with a global coupling. The phase space $X$ of such dynamical systems is the direct product of $N\geq 2$ copies of a real interval $I$, that is, $X=I^N$. An orbit 
$\{x^t\}_{t\in\N}\in X^{\N}$ is given by    
\begin{equation}\label{CML}
x_i^{t+1}=(1-\varepsilon)f(x_i^t)+\frac{\varepsilon}{N} \sum_{j=1}^N f(x_j^t) \quad\forall i\in\{1,2,\dots, N\},
\end{equation}
where $f:I\to I$ is a real map called the local dynamics and $\varepsilon\in[0,1]$ is a parameter standing for the coupling intensity. The orbits can also be obtained by iterations of the map $F_\varepsilon:X \to X$ defined by  
\[
F_\varepsilon (x)=C_\varepsilon \circ F_0(x),
\]
where $F_0(x)=(f(x_1),f(x_2),\dots,f(x_N))$ and $C_\epsilon$ is the linear coupling operator whose associated matrix is

\begin{equation*}
\begin{pmatrix}
1-\varepsilon+\frac{\varepsilon}{N} & \frac{\varepsilon}N &\dots & \dots & \frac{\varepsilon}N \\

\frac{\varepsilon}N & 1-\varepsilon+\frac{\varepsilon}{N} & \frac{\varepsilon}N & \dots& \frac{\varepsilon}N\\

\vdots &  & \ddots  &  & \vdots \\

\vdots &  &  & \ddots & \vdots&\\

\frac{\varepsilon}N & \dots & \dots & \frac{\varepsilon}N & 1-\varepsilon+\frac{\varepsilon}{N} \\
\end{pmatrix}.
\end{equation*}


\begin{definition} 
\noindent 1) We say that  $x\in X$  is a cluster state if it belongs to a cluster space
\[
\mathcal{C}_J:=\{x\in X : x_i=x_{j}, \forall i,j\in J\},
\]
for some $J\subset\{1,\dots,N\}$ such that $\#J>1$.\\ 

\noindent 2) We say that $\chi_J$ is a chimera space of size $2\leq k\leq N-1$, if  there exists $J$ such that $\#J=k$ and $\chi_J:=\mathcal{C}_J\setminus\mathcal{C}_{J'}$
for any $J'$ such that $\#J'>\#J$. A point $x\in\chi_J$ is called a chimera state.\\
\noindent 3) If $x$ belongs to the diagonal $\mathcal{D}:=\mathcal{C}_{\{1,2,\dots,N\}}$, we say that $x$ is fully synchronized.
\end{definition}

In other words, $x\in \mathcal{C}_J$ if at least 
the sites labeled by $J$ are synchronized. On the other hand, $x\in \chi_J$ if the sites labeled by $J$ are synchronized, but no other sites are synchronized with those of $J$. Note that in our definition, the diagonal is not considered a chimera space ($k<N$). We can check that any cluster space $\mathcal{C}_J$ is forward invariant by the dynamics. Therefore, once the sites labeled by the elements of $J$ are synchronized, they remain synchronized forever. However, if $x\in\mathcal{C}_J$, it can be mapped in $\mathcal{C}_{J'}$, where $J\subset J'$, and more sites can eventually synchronize with those of $J$. It follows that a chimera state can evolve toward a chimera state of larger size, or even a completely synchronized state. Thus, a chimera space is not necessarily forward invariant. In what follows, we will call a chimeric orbit, an orbit contained in $\chi_J$, for some $J\subset \{1,...,N \}$. For such an orbit, the sites labelled by $J$
are always synchronized, while the remaining sites never get synchronized with those of $J$.

In this paper, we are interested in the existence of a chimera space attracting a large set of initial conditions in $X$. To be more precise, we denote $\omega(x)$ the $\omega$-limit set of a point $x\in X$, that is, the set of the accumulation points of the orbit $\{F^n_\varepsilon(x)\}_{n\in\N}$, and denote
\[
\mathcal{B}(\chi_J):=\{x\in X : \omega(x)\subset \chi_J\}
\]
the basin of attraction of $\chi_J$.
We ask for the existence of a set $J\subset\{1,\dots N\}$ such that  $\mathcal{B}(\chi_J)$
has positive Lebesgue measure, for some values of the parameters of the system \eqref{CML}. In such a case, there is a positive probability for a random initial condition to be attracted to $\chi_J$. In particular, we may be able to observe the asymptotic synchronization of the sites of the  $J$ (and only those of $J$) in a numerical experiment. We note that a sufficient condition for $\mathcal{B}(\chi_J)$ having  positive Lebesgue measure is the existence of a Milnor attractor in $\chi_J$.

\section{Transverse Lyapunov exponent in cluster spaces}\label{stab}

Transverse Lyapunov exponents are a tool to assess the stability of invariant sets. Since chimera spaces are not necessarily invariant by the dynamics, we will instead introduce and compute Lyapunov exponents transverse to a cluster space $\mathcal{C}_J$ that contains the chimera space $\chi_J$. We start by noticing that, given the symmetries of the system \eqref{CML}, the study of the properties of the dynamics in any cluster space $\mathcal{C}_J$ can be done without loss of generality  in the cluster space
\[
\mathcal{C}_k:=\{x\in X : x_1=x_{2}=\dots=x_{k}\}, 
\]
where $k=\#J\in\{2,\dots, N\}$.

The Lyapunov exponent of $x\in \mathcal{C}_k$ in the direction $\xi\in\RR^N\setminus\{0\}$ is defined as
\begin{equation}\label{ple}
\Lambda(x,\xi):=\lim_{n\to\infty}\frac{1}{n}\log\frac{||D_{x}F^n_\varepsilon\xi||}{||\xi||},
\end{equation}
where $D_xF^n_\varepsilon$ is the Jacobian matrix of $F^n_\varepsilon$ at the point $x$. If there exists an ergodic invariant measure $\nu$ with support contained in $\mathcal{C}_k$, then the limit exists for $\nu$-almost all $x\in\mathcal{C}_k$ and any $\xi\in\RR^N$. Moreover, this limit can take at most $N$ different values as $\xi$ varies in $\RR^N$.  
Here, we are mostly interested in the Lyapunov exponents in the directions transverse to the cluster space. So let us decompose the vector space $\mathbb{R}^N$ in order to identify the directions orthogonal to $\mathcal{C}_k$. The set $\mathcal{C}_k$ can be seen as a subset of the vector space
\[ 
V_k^{\parallel}:=\text{span}(v^{k},v^{k+1},\dots,v^N), 
\]
spanned by the $N-k+1$ orthogonal vectors $v^{k},v^{k+1},\dots,v^N$ of $\mathbb{R}^N$, where $v^k$ is defined by
\[
v_1^k=v_2^k=\dots=v_k^k=1\quad\text{and}\quad v_i^k=0\quad\forall i\in\{k+1,\dots,N\},
\]
and $v^l$ is defined for any $l\in\{k+1,\dots,N\}$ by 
\[
v^l_l=1 \quad\text{and}\quad v_i^l=0\quad\forall i\in\{1,\dots,N\}\setminus\{l\}.
\]
 Denoting $V^{\bot}_{k}$ the orthogonal complement of $V_k^{\parallel}$ in $\mathbb{R}^N$ with respect to the canonical inner product, we have
\begin{equation}\label{deco}
\mathbb{R}^N= V_k^{\parallel}\oplus V^{\bot}_{k}.
\end{equation}
For a vector $\xi\in\RR^N$, we denote by $\xi_\bot$ and by $\xi_\parallel$ be the projections of $\xi$ on $V^\bot_{k}$ and on $V_k^{\parallel}$, respectively.

\begin{definition}\label{tly}
Let $k\geq 2$, $\xi \in \mathbb{R}^N\setminus V_k^{\parallel} $ and $x\in\mathcal{C}_k$ be such that $F_\varepsilon$ is differentiable along the orbit of $x$. If the limit exists, 
\begin{equation}
\Lambda^\bot_k(x,\xi)=  \lim_{n\to \infty} \frac1n \log \frac{||(D_{x}F^n_\varepsilon\xi)_\bot||}{||\xi_\bot||},
\end{equation}
is called the transverse Lyapunov exponent (TLE) of $x\in \mathcal{C}_k$ in the direction $\xi$.
\end{definition}
The TLE measures the asymptotic exponential expansion rate of a vector in the directions orthogonal to $\mathcal{C}_k$ under the iterations of the linearized dynamics. We will see that, in our case, if $\Lambda^\bot_k(x,\xi)$ exists, it does not depend on the choice of $\xi\in \mathbb{R}^N\setminus V_k^{\parallel} $. So, we will simply write $\Lambda^\bot_k(x)$ and we have $\Lambda^\bot_k(x)=\Lambda(x,\xi)$ for all $x\in\mathcal{C}_k$ and $\xi\in V_k^\bot\setminus\{0\}$.

The limit defining the TLE may not always exist. In this case, we denote with $\Lambda^{\bot,-}_k(x)$ and $\Lambda^{\bot,+}_k(x)$, respectively, the inferior limit and superior limit. 

Now, we give the expression of the TLE of cluster states.

\begin{theorem}\label{the1}
Let $k\geq 2$ and  $x\in \mathcal{C}_k$ such that $F_\varepsilon$ is differentiable along the orbit $\{x^t\}_{t\in\N}$ of $x$. Then, if it exists,  
\begin{equation}\label{lam}
\Lambda^\bot_k(x)= \log(1-\varepsilon)+\lim_{n\to \infty} \frac1n \sum_{t=0}^{n-1}\log |f'(x_1^t)|.
\end{equation}
\end{theorem}

\begin{proof} 
For any $x\in \mathcal{C}_k$ such that $F_\varepsilon$ is differentiable at $x$, we have that 
\begin{equation*}
D_{x}F_\varepsilon=\begin{pmatrix}
af'(x_1) & bf'(x_1) &\dots & bf'(x_1)& bf'(x_{k+1})& \dots & \dots &bf'(x_N) \\
bf'(x_1) & af'(x_1) &\ddots& \vdots  & \vdots  & \dots & \dots &\vdots \\
\vdots   & bf'(x_1) &\ddots& bf'(x_1)  & \vdots  & \dots & \dots &\vdots \\
\vdots   & \vdots   &\ddots& af'(x_1)& bf'(x_{k+1})& \dots & \dots &\vdots \\
\vdots   & \vdots   &\dots & bf'(x_1)& af'(x_{k+1})& \dots & \dots &\vdots \\

\vdots   & \vdots   &\dots & \vdots & bf'(x_{k+1})& \ddots & \dots &\vdots \\
\vdots   & \vdots   &\dots & \vdots & \vdots & \dots  & \ddots &bf'(x_N)\\
bf'(x_1) & bf'(x_1) &\dots & bf'(x_1)& bf'(x_{k+1})& \dots &\dots & af'(x_N) \\
\end{pmatrix},
\end{equation*}
where
\[
a=1-\varepsilon+\frac{\varepsilon}{N} \quad\text{and}\quad\quad b=\frac{\varepsilon}{N}.
\]

Let $v^1,v^2,\dots,v^{k-1}$ be  the $k-1$ vectors of 
$\mathbb{R}^N$ defined for every $l\in\{1,\dots,k-1\}$ by
\[
v_l^l=1,\ v_{l+1}^l=-1\quad\text{and}\quad v_i^l=0\quad\forall i\in\{1,\dots,N\}\setminus\{l,l+1\}.
\] 
Then,  one can check that $v^1,v^2,\dots,v^{k-1}$  are linearly independent eigenvectors of $D_{x}F_\varepsilon$ with eigenvalue 
\[
\lambda_x=f'(x_1)(1-\varepsilon).
\]
Moreover,  these vectors are orthogonal to  the vectors $v^k,v^{k+1},\dots,v^{N}$ that span $V^{\parallel}_{k}$. It follows that $v^1,v^2,\dots,v^N$ is a basis of $\mathbb{R}^N$ and that  
\begin{equation}
\text{span}(v^{1},\dots,v^{k-1}) = V^{\bot}_{k}.
\end{equation}
We note that both spaces $V^{\parallel}_{k}$ and $V^{\bot}_{k}$ do not depend on $x$ and are stable by the action of $D_{x}F_\varepsilon$ for any $x\in\mathcal{C}_k$. Then, if $x\in\mathcal{C}_k$, we have the following decomposition of $D_{x}F_\varepsilon\xi$ in $V_k^{\parallel}$ and $V^{\bot}_{k}$:
\[
(D_{x}F_\varepsilon\xi)_\parallel=D_{x}F_\varepsilon \xi_\parallel\qquad\text{and}\qquad 
(D_{x}F_\varepsilon\xi)_\bot=D_{x}F_\varepsilon\xi_\bot=\lambda_x\xi_\bot.
\]
 Now, if $\{x^t\}_{t\in\N}$ is the orbit of a point $x\in\mathcal{C}_k$, then $x^t$ belongs to $\mathcal{C}_k$ 
 for any $t\in\N$ and we have
\begin{equation}
\begin{aligned}
D_{x}F^n_\varepsilon\xi=\prod_{t=0}^{n-1}D_{x^{t}}F_\varepsilon(\xi_\parallel+\xi_\bot) &= 
\prod_{t=0}^{n-1}D_{x^{t}}F_\varepsilon\xi_\parallel+\prod_{t=0}^{n-1}D_{x^{t}}F_\varepsilon\xi_\bot,
\end{aligned}
\end{equation}
which implies that 

\begin{equation}
\begin{aligned}
(D_{x}F^n_\varepsilon\xi)_\bot=\xi_{\bot}\prod_{t=0}^{n-1} \lambda_{x^t}  =
(1-\varepsilon)^n  \xi_{\bot}\prod_{t=0}^{n-1} f'(x_1^t).
\end{aligned}
\end{equation}
Then, formula \eqref{lam} follows from Definition \ref{tly}.
\end{proof}

\begin{remark}\label{erg} If we introduce the projection $\pi_1:X\to I$ defined by $\pi_1(x)=x_1$, we obtain
\[
\Lambda^\bot_k(x)= \log(1-\varepsilon)+\lim_{n\to \infty} \frac1n \sum_{t=0}^{n-1}\log |f'\circ\pi_1 (F_\varepsilon^t(x))|\qquad\forall x\in\mathcal{C}_k.
\]
\noindent 1) If $x$ belongs to the diagonal $\mathcal{D}$ then we can check from \eqref{CML} that $x_1^{t}=f^t(x_1)$ for all $t\in\NN$. It follows that 
\[
\Lambda^\bot_k(x)= \log(1-\varepsilon)+\lim_{n\to \infty} \frac1n \sum_{t=0}^{n-1}\log |f'(f^t(x_1))|\qquad\forall x\in\mathcal{D}.
\]
As $x_1\in I$, if $f$ admits an ergodic invariant measure $\mu$ and $\log |f'|$ belongs to $L^1_\mu$, then for $\mu$ almost all $x_1\in I$, we have
\begin{equation}\label{tlediag}
\Lambda^\bot_k(x)= \log(1-\varepsilon)+\lambda_\mu
\qquad\text{where}\quad
\lambda_\mu=\int_I \log |f'(z)|d\mu(z)
\end{equation}

is the Lyapunov exponent of the map $f$ with respect to $\mu$.\\

\noindent 2) Now, if $x\in\mathcal{C}_k\setminus\mathcal{D}$, then $\pi_1(F_\varepsilon^t(x))$ is not anymore an orbit of the local map, but the TLE is given by a Birkhoff sum of the observable $\log |f'\circ\pi_1|$ along an orbit of $F_\varepsilon$. Therefore, the existence and the values of the TLE may now depend on the existence and the properties of ergodic invariant measures of $F_\varepsilon$ supported in $\mathcal{C}_k$.\\

\noindent 3) Note that if $\varepsilon>1-1/\sup|f'|$, then for every $k\in\{2,\dots, N\}$, we have $\Lambda_k^{\bot,+}(x)<0$ for any $x\in\mathcal{C}_k$ such that $F_\varepsilon$ is differentiable along the orbit of $x$. 
Conversely,  if $\varepsilon<1-1/\inf|f'|$, then for every $k\in\{2,\dots, N\}$, we have $\Lambda_k^{\bot,-}(x)>0$ for any $x$ satisfying the same hypothesis as above.
\end{remark}

As mentioned earlier,  we are interested in the existence of sets of positive Lebesgue measure attracted by chimera spaces. Since any cluster space $\mathcal{C}$ is forward invariant, the restriction $F_{\varepsilon|\mathcal{C}}:\mathcal{C}\to\mathcal{C}$ of $F_{\varepsilon}$ to $\mathcal{C}$ defines  a dynamical system with its own attractors. 
Suppose $A\subset\mathcal{C}$ is such an attractor. Then, a general result from Alexander et al. \cite{alexander} states that  if $A$ supports a strong SRB-measure and that the TLE computed with respect to this measure is negative, then $A$ is a Milnor attractor for the whole system $F_\varepsilon: X\to X$. An other general result, that can be found in \cite{ashwin}, is that if $A$ is asymptotically stable and the TLE is negative for any ergodic measures of the restrained dynamics, then $A$ is asymptotically stable for the whole system, under some suitable assumptions on the derivative of the map. Both results imply that $\mathcal{B}(A)$, and therefore $\mathcal{B}(\mathcal{C})$, has positive Lebesgue measure in the whole space. To sum up, if $F_{\varepsilon|\mathcal{C}}:\mathcal{C}\to\mathcal{C}$ and one of its attractor $A$ satisfy the  hypothesis mentioned above of \cite{alexander} or \cite{ashwin}, the existence of a set of positive Lebesgue measure attracted by the cluster space $\mathcal{C}$ can be proved studying the sign of the TLE associated with the different ergodic measures supported in $A$.

Unfortunately, for chaotic local maps the existence of ergodic invariant measures with support contained in cluster spaces (excluding the diagonal) is principally known for weak coupling, for which the TLE is  close to the (positive) Lyapunov exponent of the local map. Also, the explicit computation of the TLE in cluster spaces is challenging because it depends on the orbits of the CML and not only on those of the local dynamics as it is the case in the diagonal (see Remark \ref{erg}). Nevertheless, the results of \cite{alexander, ashwin} motivate the use of the TLE as an indicator to detect possible attracting chimeras. Also, formula \eqref{lam}, is quite suitable for numerical computation. This is what is done in the next section, where we perform a numerical exploration of the TLE and present numerical examples of attracting chimeras of different types. This suggests, that either the CML satisfies the different hypothesis of \cite{alexander} or \cite{ashwin}, or that these results could also extend to cover such CMLs.

\section{Numerical hunt for chimeras}

Let us consider the CML \eqref{CML} with local dynamics defined in $I=[0, 1]$ by the logistic map
\begin{equation}\label{logi}
f(x)=4x(1-x).
\end{equation}
This map admits a unique absolutely continuous invariant measure $\mu$ and a positive associated Lyapunov exponent $\lambda_\mu=\log 2$. We will perform our numerical study with this local map, but we expect qualitatively similar behavior for different kinds of local maps, in particular, other types of unimodal maps and tent maps that are topologically conjugated to $f$.

\subsection{Bifurcation diagrams and TLEs} 
Let us for now consider $N=3$ sites and study the existence of attracting chimeras spaces of sizes $k=2$. To that aim, we perform numerical simulations of the long-term dynamics of points in $\mathcal{C}_2$ (for which $x_1=x_2$) and estimate their associated TLE for a wide range of coupling values. 

\begin{figure}
\includegraphics[height=3in]{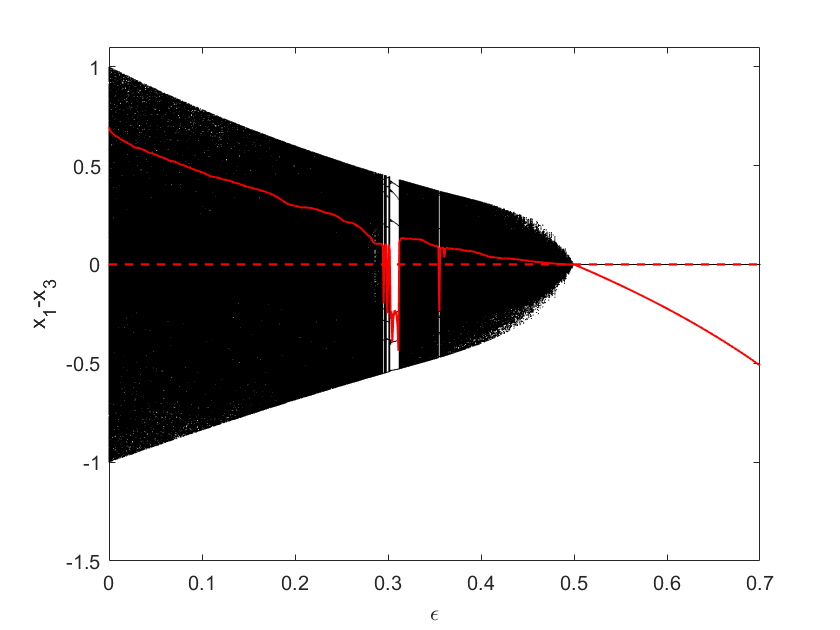}
\caption{TLE (red) of a random initial condition $x=(x_1,x_1,x_3)$ and the quantity $x^t_1-x^t_3$ (black) for the orbit of $x$ at time $t\geq 10^7$ and $\varepsilon\in[0,0.7]$.}\label{f2}
\end{figure}

For 750 uniformly spaced values of $\varepsilon\in [0, 0.7]$, we chose at random a point $x=(x_1,x_1,x_3)\in \mathcal{C}_2$ and plot on the vertical axis of Figure \ref{f2} the value of $x^t_1-x^t_3$ for $2.10^3$ consecutive values of $t\geq 10^7$. The TLE associated with $x$ is estimated by computing the Birkhoff sum in formula \eqref{lam} up to a time $n$ for which the last 3 successive values of the sum are within the same interval of size $10^{-6}$. For all tested $\varepsilon$, and for $100$ different random choices of $x$, we always obtain the same estimation of the TLE, and observe the same qualitative behavior of the orbit.

As $\varepsilon$ increases, the length of the interval containing  $x_1^t-x_3^t$ for large $t$ decreases, indicating that the limit orbit of $x=(x_1,x_1,x_3)$ becomes closer to the diagonal. In particular, for $\varepsilon>1-\exp(-\lambda_\mu)=0.5$, according to \eqref{tlediag}, the TLE is negative for $\mu$ almost all point of the diagonal, and all the tested starting points in $\mathcal{C}_2$ end up in a fully synchronized state. Note that the diagonal is not expected to attract all of the initial conditions in $X$. Its basin of attraction could for instance have a riddled structure, if there exists another ergodic invariant measure $\mu'$ of the local map such that $\varepsilon<1-\exp(-\lambda_{\mu'})$ \cite{alexander,ashwin}.
On the other hand, according to the point $3)$ in Remark \ref{erg}, for $\varepsilon>1-1/\sup|f'|=3/4$, we have $\Lambda_N^{\bot,+}(x)<0$ for every point $x$ of the diagonal. 

For $\varepsilon \in [0, 0.5]$, the TLE is mostly  positive and $\{x^t_1-x^t_3\}_{t\in\N}$ appears to be dense in a sub-domain of $[-1,1]$. Yet, for some values of $\varepsilon$ in the interval $[0.29, 0.37]$, we notice some abrupt changes in the sign of the TLE, associated with several bifurcations. We explore more in detail this region in Figure \ref{fbifzoom}, where we distinguish 3 different types of regimes:\\

\begin{figure}[t!]
    \centering
    \begin{subfigure}[t]{0.5\textwidth}
        \centering
        \includegraphics[height=2.5in]{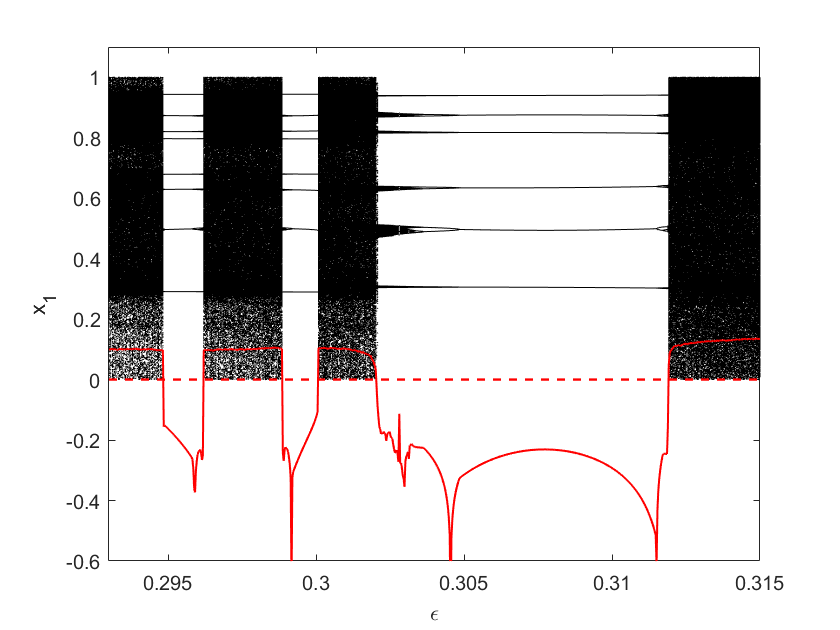}
        \caption{}
    \end{subfigure}
    ~ 
    \begin{subfigure}[t]{0.5\textwidth}
        \centering
        \includegraphics[height=2.5in]{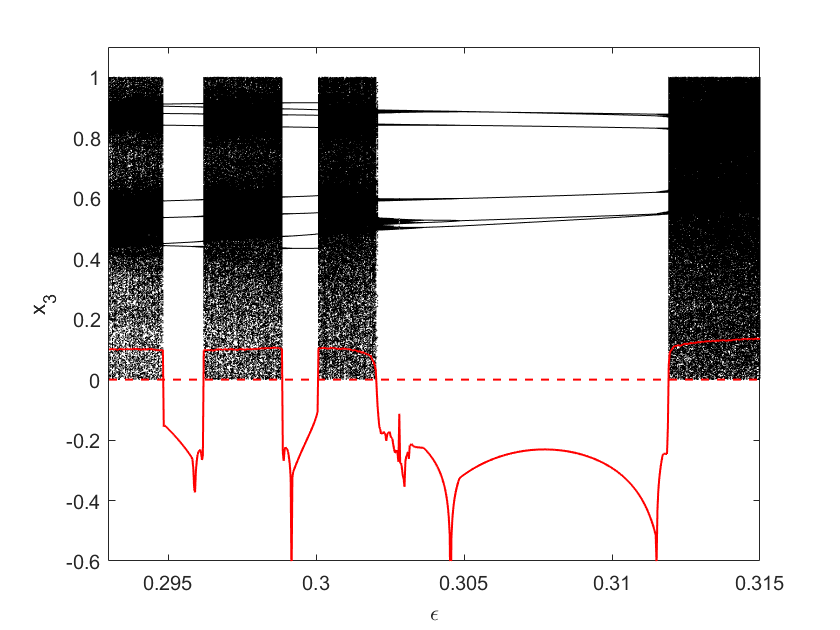}
        \caption{}
    \end{subfigure}
    \caption{TLE (red) of a random initial condition $x=(x_1,x_1,x_3)$ and the values of $x^t_1$ (A) and $x^t_3$ (B) for the orbit of $x$ at time $t\geq 10^7$ and $\varepsilon\in[0.293,0.315]$.}
    \label{fbifzoom}
\end{figure}

{\bf 1) Repelling chimeras with dense dynamics.}
For the values of $\varepsilon$ where the TLE is positive, both $\{x^t_1\}_{t\in\N}$ and $\{x^t_3\}_{t\in\N}$ appear to be dense in $[0,1]$. In Figure \ref{f18}, we plot the corresponding attractor for $\varepsilon=0.301$. For these regimes, numerical estimates of the two Lyapunov exponents in the directions parallel to $\mathcal{C}_2$ ($\xi\in V_2^{\parallel}$ in \eqref{ple}) give positive values which are independent on the random starting point. This suggests that $\mathcal{C}_2$ hosts an absolutely continuous ergodic invariant measure, as shown by the empirical density plotted in Figure  \ref{f18}, for $\varepsilon=0.301$. We observe some singularities, probably originating from the singularities of the invariant density of the local map $f$ at 0 and 1. The positive TLE suggests  that $\mathcal{C}_2$ does not attract generic points in $X\setminus \mathcal{C}_2$ in these regimes, which is confirmed by our simulations.\\

\begin{figure*}[t!]
    \centering

    \begin{subfigure}[t]{0.5\textwidth}
        \centering
        \includegraphics[height=2.1in]{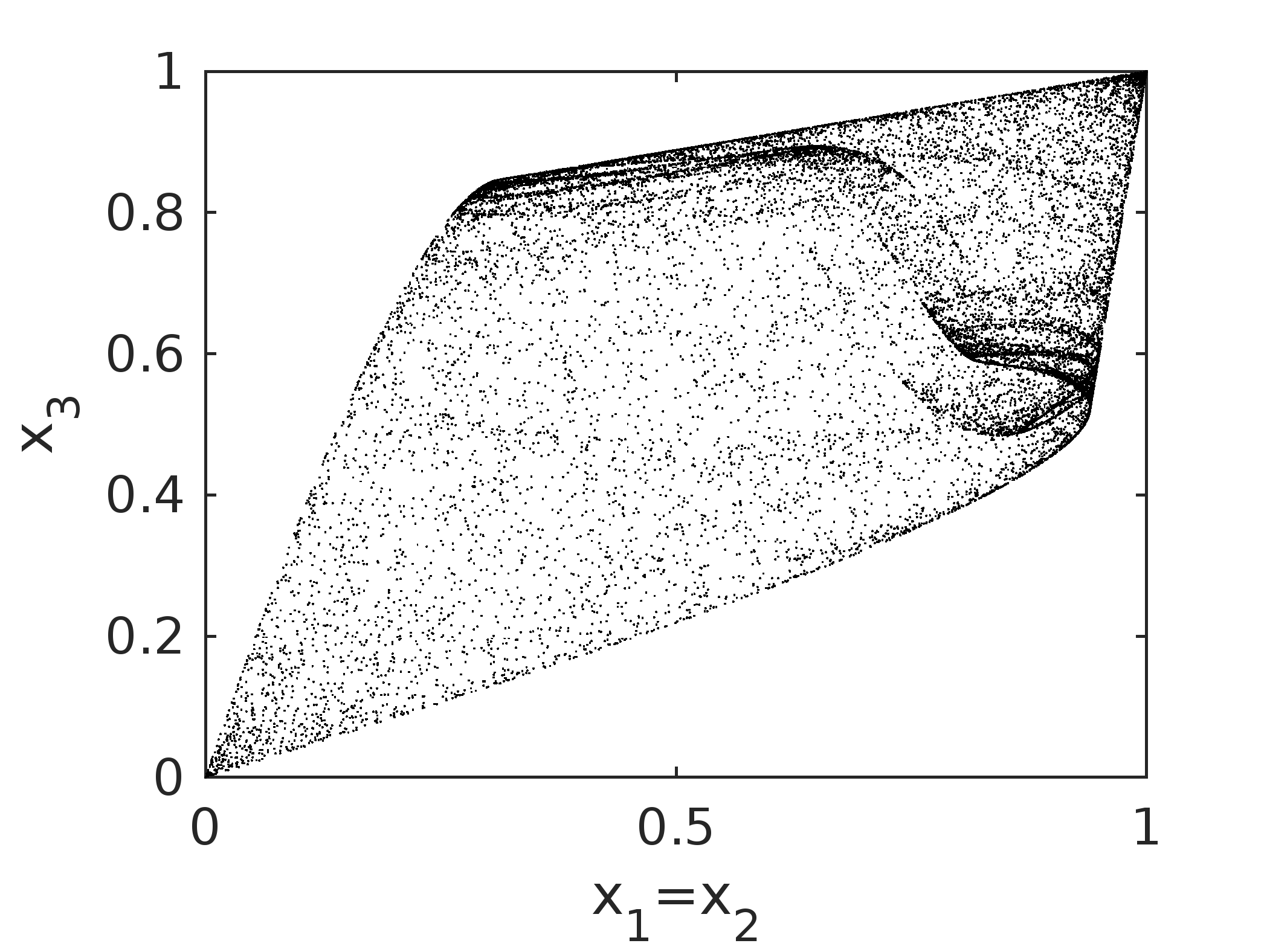}
        \caption{}
    \end{subfigure}
 
    \begin{subfigure}[t]{0.5\textwidth}
        \centering
        \includegraphics[height=2.1in]{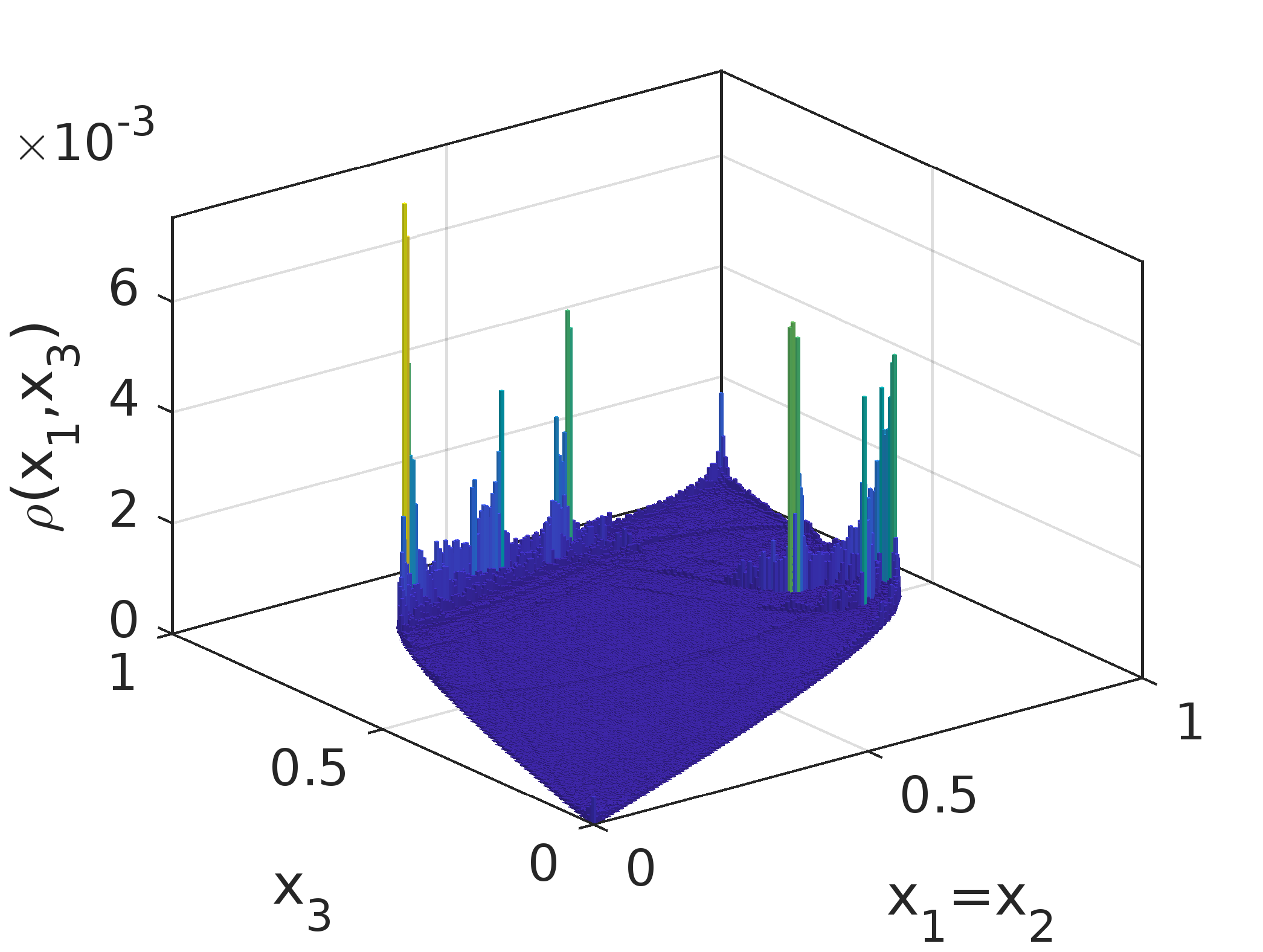}
        \caption{}
    \end{subfigure}%

    \caption{Attractor (A) of the orbits of $\mathcal{C}_2$ for $\varepsilon=0.301$ and its associated empirical density (B).}\label{f18}
\end{figure*}

\begin{figure}
\includegraphics[height=2.5in]{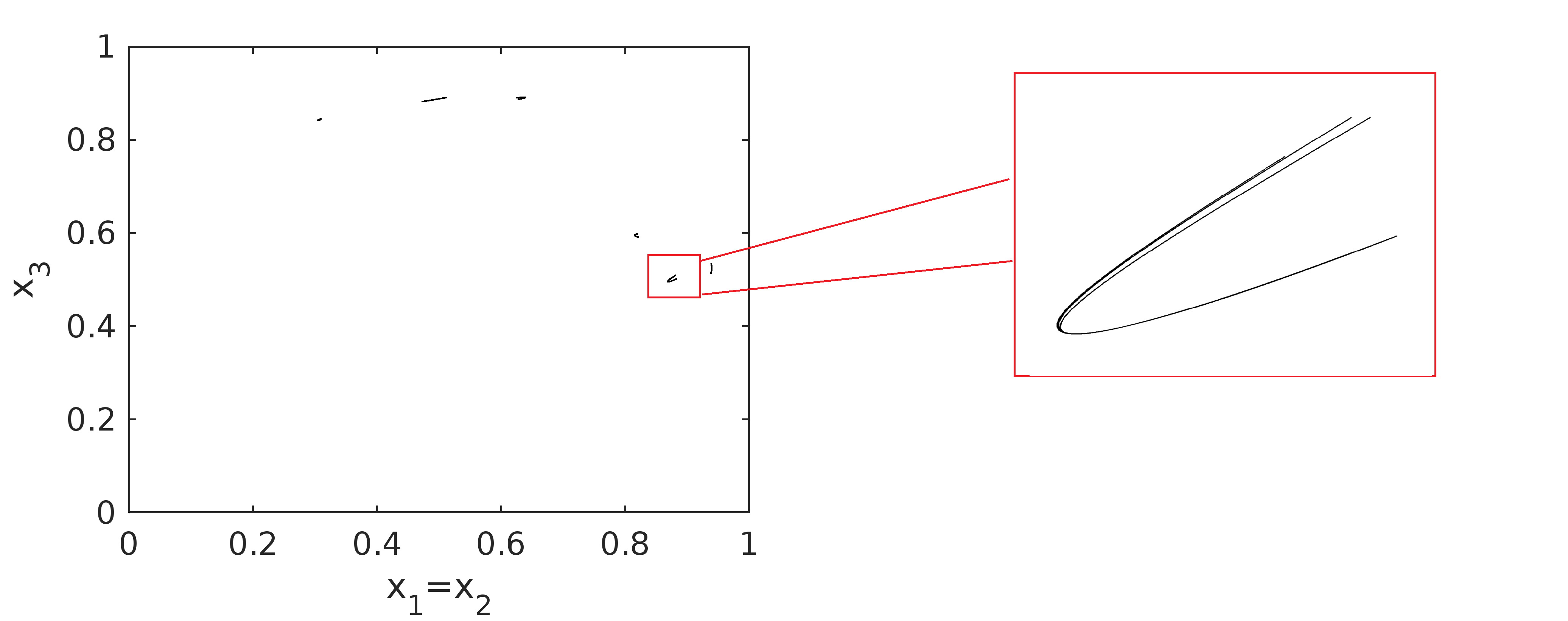}
\caption{Attractor of the orbits of $\mathcal{C}_2$ for $\varepsilon=0.30215$.}\label{f12}
\end{figure}

{\bf 2) Attracting chimeras with hyperbolic attractors.}
As shown in Figure \ref{fbifzoom}, the dynamics in $\mathcal{C}_2$ undergoes several bifurcations. For $\varepsilon\approx 0.302$ for instance, the TLE becomes negative. For a range of coupling near this value, we observe that a collection of Hénon-type sets contained in $\chi_{\{1,2\}}\subset\mathcal{C}_2$ attracts the orbits of $\mathcal{C}_2$ (see Figure \ref{f12}). Also, the two Lyapunov exponents in the directions parallel to $\mathcal{C}_2$ are of opposite signs, which indicates that the dynamics inside $\mathcal{C}_2$ is hyperbolic. As suggested by the negative sign of the TLE, all of the tested random initial conditions in $X$ ended up attracted to the Hénon-type sets contained in $\chi_{\{1,2\}}$ or to their symmetric analogs contained in $\chi_{\{1,3\}}$ or $\chi_{\{2,3\}}$. Overall, these simulations support the existence of chimera spaces containing a chaotic attractor and whose basin of attraction has positive Lebesgue measure in $X$.\\

{\bf 3) Attracting periodic cycles.} For $\varepsilon\simeq 0.3045$ and $\varepsilon\simeq 0.3115$, the TLE diverges and the system undergoes a bifurcation. Between this two values, the dynamics in $\mathcal{C}_2$ admits a limit cycle. For instance, for $\varepsilon=0.306$, the cycle is of period 6 and contains the point
\begin{equation}
\begin{aligned}
x=(0.494797626302772,   0.494797626302772,   0.884647247995353).\\
\end{aligned}
\end{equation}
We show in Figure \ref{cycle} the projection of this cycle in the plane $x_1,x_3$. Its associated TLE is $$\Lambda^{\bot}_2(x)=-0.2620197...$$ and is negative, due to the closeness of the first coordinate of $x$ to the critical point of $f$. The spectral radius of the matrix $DF^6_\varepsilon(x)$ is smaller than 1, indicating that this cycle of chimera states is an  asymptotically stable attractor of $F_\varepsilon:X\to X$. Therefore, it attracts a set of positive Lebesgue measure in $X$. In fact, this cycle, together with the two other cycles obtained by permutation of the variables $x_1,x_2$ and $x_3$, attract all of the tested initial conditions in $X$\footnote{The basin of attraction of these cycles is not the whole phase-space $X$. It is clear for instance that the diagonal and its preimages do not belong to this basin.}. 

\begin{figure}
\includegraphics[height=3in]{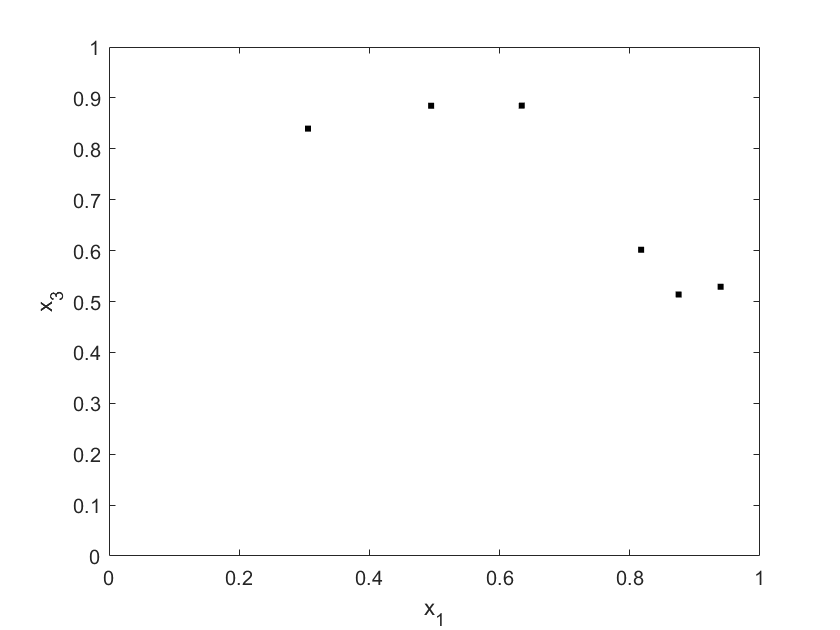}
\caption{Projection in the plane $x_1,x_3$ of the limit cycle contained in $\mathcal{C}_2$ for $\varepsilon=0.306$.}\label{cycle}
\end{figure}

Considering the symmetries of the CML, one could wonder if cycles are likely to belong to a cluster space. The following result provides the beginning of an answer.

\begin{proposition}
Suppose $f$ is a polynomial of degree $d$ and let $c$ be a cycle of period $p \ge 1$ of $F_{\varepsilon}$. 
If 
\begin{equation}\label{crit}
p<\frac{\log_d N!}{N},
\end{equation}
then there exists $k\geq 2$ such that $c\subset \mathcal{C}_k$.
\end{proposition}

\begin{proof}
If $f$ is a polynomial of degree $d$, then 
\begin{equation}\label{fix}
F_{\varepsilon}^p(x)=x,
\end{equation} 
 is a system of $N$ polynomial equations of degree $d^p$ with $N$ variables. By Bézout's Theorem, this kind of systems have at most ${d^{Np}}$ solutions (see for instance Theorem 4 in \cite{bezout}). Now, let $x$ be a solution of \eqref{fix}. Then, any permutation $\sigma(x)$ of the coordinates of $x$ is also a solution, since for our CML we have $\sigma\circ F_\varepsilon=F_\varepsilon\circ\sigma$. 
If $x_i\neq x_j$ for all $i\neq j\in\{1,\dots,N\}$, then $\sigma(x)\neq x$ for every permutation $\sigma(x)$. Therefore, \eqref{fix} has at least $N!$ different solutions and $d^{Np}\geq N!$, which contradicts \eqref{crit}. Thus, if $x$ is a periodic point of period $p$ and \eqref{crit} holds, there exist
$i\neq j\in\{1,\dots,N\}$ such that $x_i=x_j$. In other words,  $x\in\mathcal{C}_k$ for some $k\geq 2$ and its orbit $c$ is contained in $\mathcal{C}_k$.
\end{proof}

\begin{remark}
\noindent 1) Since the sequence $g_N:=\frac{\log_d N!}{N}$ tends to infinity, for any polynomial $f$ and  $p$ fixed, every cycle of period $p$ belongs to a clustered space, provided $N$ is large enough. For the logistic map ($d=2$), we can show for instance that every cycle of period $p\le 2$ belongs to a cluster space provided $N>8$. Note that, by Stirling formula, $g_N$ can be estimated for large $N$ by $$g_N\underset{N\to \infty}{\sim} \log_d(N).$$ 

\noindent 2) We expect criteria sharper than (\ref{crit}) to exist. In particular, the number of solutions of (\ref{fix}) which belong to $[0,1]^N$ should be much less than $d^{Np}$. See for example \cite{fewn}, which provides a bound for the number of positive real solutions for such systems of polynomial equations.
\end{remark}

\subsection{Chimeras for larger systems} Our observations for $N>3$, are the following: As for the case $N=3$, for the values of coupling where the TLE is positive, it does not depend on the starting point. This could indicate that each $\mathcal{C}_k$, for $k=1,...,N$ hosts a unique absolutely continuous invariant measure. On the other hand, when the TLE is negative, its value can depend on the starting point, and the attractors of $F_\varepsilon|_{\mathcal{C}_k}$ are either cycles or fractal sets. 

For instance, for $N=5$, $k=2$ and $\varepsilon\in(0.235,0.319)$, the TLE is negative and takes two different values depending on the initial condition (see Figure \ref{TLE5}). In that case, $\mathcal{C}_2$ hosts 4 distinct attractors $A_1,A_2,A_3,A_4$ contained in $\chi_{\{1,2\}}\cap \chi_{\{3,4,5\}}$, $\chi_{\{1,2,3\}}\cap\chi_{\{4,5\}}$, $\chi_{\{1,2,4\}}\cap\chi_{\{3,5\}}$ or $\chi_{\{1,2,5\}}\cap\chi_{\{3,4\}}$, respectively. All these attractors are obtained by permuting the variables of any of the others. Depending on the size of the block in which $x_1$ ends up in, the TLE can take two different values: $\Lambda_1$ if the system is attracted to $A_1$ and $\Lambda_2$ if it is attracted to $A_2$, $A_3$ or $A_4$. We plot the attractor $A_1$ for $\varepsilon=0.2373$ in Figure \ref{att5}.

For even bigger systems, and for coupling regimes where the TLE is negative, we observe the emergence of different blocks of synchronization, whose sizes and dynamics depend on the initial condition. For $N=11$ for instance, and for some values of $\varepsilon$, depending on the starting point, the dynamics is attracted to a set of the form $\mathcal{\chi}_J\cup \mathcal{\chi}_{J'}$, where either $|J|=5$ and $|J'|=6$ or $|J|=4$ and $|J'|=7$. Up to $N=100$, we still observe regimes where the TLE is negative and the created chimeras always contain several synchronized blocks, whose sizes depend on the initial condition.\\
For $N=7$ and $\varepsilon=0.23$, we observe for a large number of initial conditions the emergence of chimeras constituted by two synchronized blocks of sizes 2 and 3 and two independent sites whose dynamics is not synchronized with any of the other sites. The TLE of each of these blocks is negative and these states admit a limit cycle of period 8. This is to our knowledge the first example of such chimeras having two sites with independent dynamics for globally coupled logistic maps for $N>3$.

\begin{figure}
\includegraphics[height=3in]{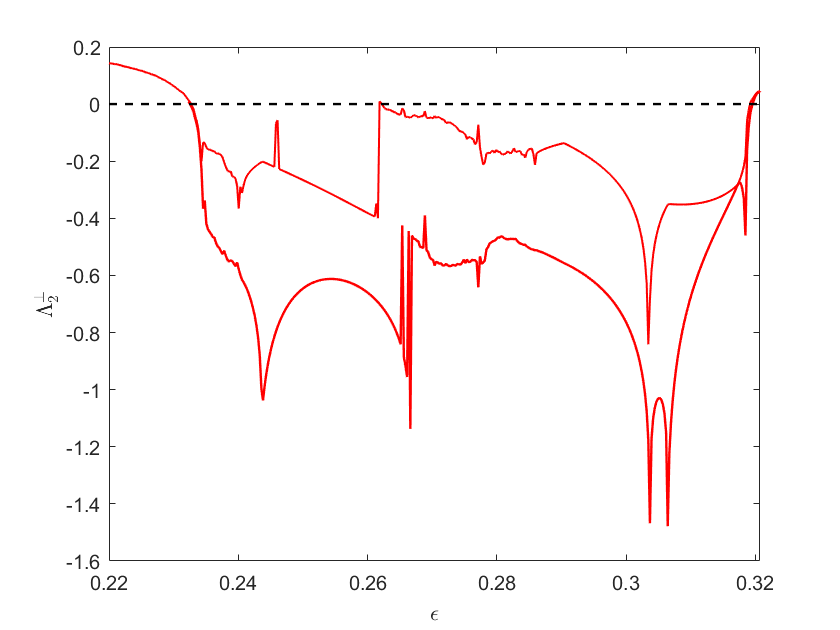}
\caption{Values of TLE of 1000 random  points of $\mathcal{C}_2$ for $N=5$ and $\varepsilon\in [0.22,0.32]$.}\label{TLE5}
\end{figure}

\begin{figure}
\includegraphics[height=3in]{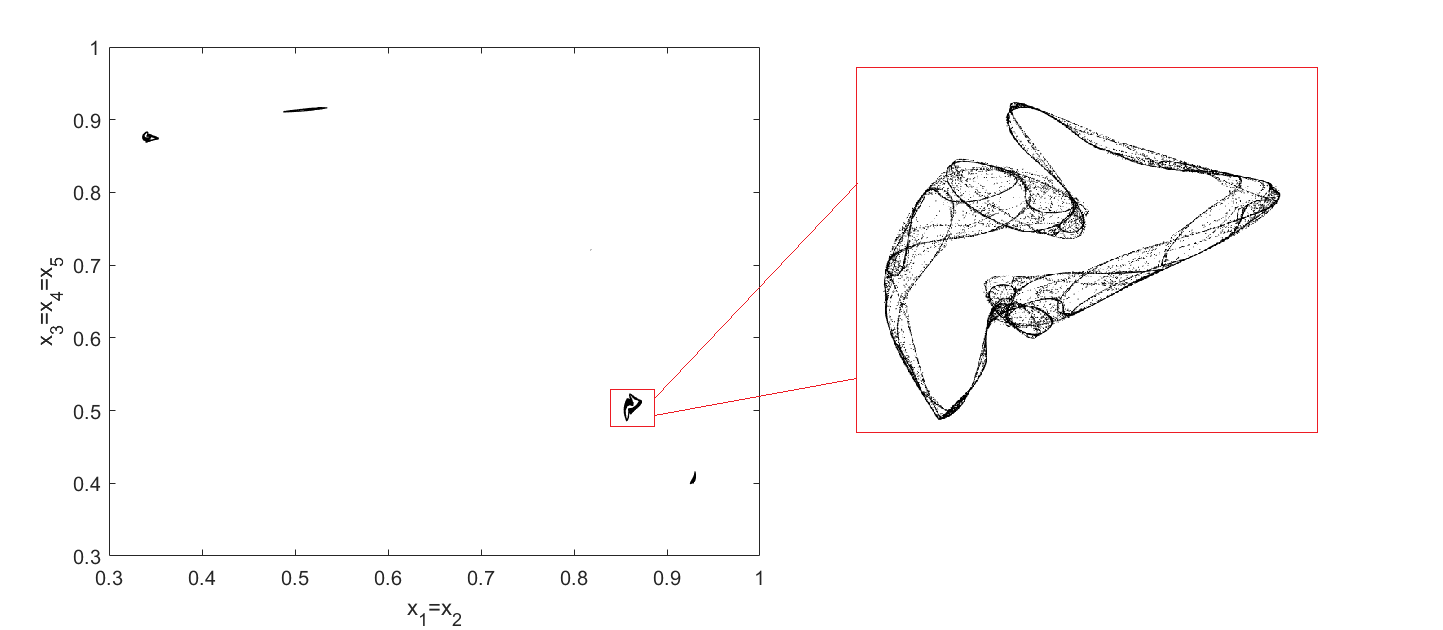}
\caption{Projection in the plane $(x_1,x_3)$ of one of the  attractors of the CML for $N=5$ and $\varepsilon=0.2373$.}\label{att5}
\end{figure}

\section{On the attractors of chimera states}

 The numerical simulations of the last section show that, in some regimes where the TLE is positive, the asymptotic dynamics in cluster spaces can take place on a large domain, suggesting the existence of absolutely continuous invariant measures for this dynamics (see Figure \ref{f18}). 
 However, we did not observe chimeras state with a negative TLE and whose asymptotic dynamics is supported on a set of positive Lebesgue measure. We rather observe limit cycles or fractal attractors. We then wonder if chimera spaces with negative TLE (and thus attracting) can host attractors of positive Lebesgue measure. In this section, we show that this cannot happen if the coupling is strong enough. To this end, let us study the dynamics inside cluster spaces.
 
As mentioned before, any chimera space of size $k$ belongs to a forward invariant subspace 
\[
\mathcal{C}_{J}=\{x\in X :  x_i=x_j,\forall i,j\in J\},
\]
for some $J\subset\{1,2,\dots,N\}$ of cardinality $k$. When restricted to one of these subspaces, the map $F_\epsilon$ can be identified to a new CML with $N-k+1$ sites $\tilde{x}_1,\tilde{x}_2,\dots, \tilde{x}_{N-k+1}$.
This new CML defined on $\tilde{X}=I^{N-k+1}$ writes  
\begin{equation}
\tilde{F}_\varepsilon(\tilde{x})=\tilde{C}_\varepsilon \circ \tilde{F}_0 (\tilde{x}),
\end{equation}
where the new map $\tilde{F}_0 : \tilde{X} \to \tilde{X}$ and coupling operator $\tilde{C}_\varepsilon:\tilde{X} \to \tilde{X}$  obtained from \eqref{CML} 
are given by
\[
\tilde{F}_0(\tilde{x}_1,...,\tilde{x}_{N-k+1}) = (f(\tilde{x}_1),...,f(\tilde{x}_{N-k+1}))
\]
and 
\begin{equation}\label{Ce}
\tilde{C}_\varepsilon(\tilde{x}_1,...,\tilde{x}_{N-k+1})=\begin{pmatrix}
1-\varepsilon +\frac{k\varepsilon}N & \frac{\varepsilon}N &\frac{\varepsilon}N &  \dots&\frac{\varepsilon}N  \\
\frac{k\varepsilon}N & 1-\varepsilon +\frac{\varepsilon}N & \frac{\varepsilon}N &  \dots&\frac{\varepsilon}N \\
\frac{k\varepsilon}N & \frac{\varepsilon}N & 1-\varepsilon+\frac{\varepsilon}N &  &\frac{\varepsilon}N \\
\frac{k\varepsilon}N & \frac{\varepsilon}N  & \frac{\varepsilon}N & \ddots & \vdots\\
\vdots &\vdots & \vdots&  &1-\varepsilon+\frac{\varepsilon}N\\ \end{pmatrix}\begin{bmatrix}
           \tilde{x}_{1} \\
           \tilde{x}_{2} \\
           \vdots \\
           \vdots \\
           \tilde{x}_{N-k+1}
         \end{bmatrix}.
         \end{equation}

\begin{proposition}\label{the2}
If $f$ is piece-wise $C^1$ and is such that

\begin{equation}\label{cond}
\sup|f'|<(1-\varepsilon)^{\frac{k-N}{N-k+1}},
\end{equation}
then the non-wandering set $\Omega$ of $\tilde{F}_\varepsilon$ has zero $(N-k+1)-$dimensional Lebesgue measure.
\end{proposition}

\begin{proof} 
Let us consider the set 
\[
\mathcal{S}=\bigcap_{n\in\N} \tilde{F}^n_\varepsilon(\tilde{X})
\]
and let us show that
\begin{equation}\label{==}
\tilde{F}_\varepsilon (\mathcal{S})= \mathcal{S}.
\end{equation}
The inclusion $ \tilde{F}_\varepsilon (\mathcal{S})\subset\mathcal{S}$ is direct.  Now, if $y\in\mathcal{S}$, then for any $n\in\N$ there exists $x_n \in \tilde{F}_\varepsilon^n(\tilde{X})$ such that $\tilde{F}_\varepsilon(x_n)=y$. 
The set $\tilde{X}$ being compact, there exists a sub-sequence $\{x_{n_i}\}_{i\in\N}$ which converges to a point $x\in \tilde{X}$. 
Let $p\in\N$. Then, for every $i\in\N$ large enough $x_{n_i}\in \tilde{F}^{n_i}_\varepsilon(\tilde{X})\subset \tilde{F}^{p}_\varepsilon(\tilde{X})$, and by compactness $x\in \tilde{F}^{p}_\varepsilon(X)$. It follows that  $x\in \mathcal{S}$ and that  
\[
\tilde{F}_\varepsilon(x)=\lim_{i \to \infty}\tilde{F}_\varepsilon(x_{n_i})=y,
\]
by continuity of $\tilde{F}_\varepsilon$.  Thus, if $y\in S$, then $y\in\tilde{F}_\varepsilon(S)$, which ends to show  \eqref{==}.\\

As
$
\tilde{F}_\varepsilon (\mathcal{S})=\tilde{C}_\varepsilon \circ \tilde{F}_0(\mathcal{S})
$
where $\tilde{C}_\varepsilon$ is linear and invertible, we have:
\begin{equation}\label{ss}
 V(\tilde{F}_\varepsilon(\mathcal{S}))=|\det D\tilde{C}_\varepsilon| V (\tilde{F}_0(\mathcal{S})),
\end{equation}
$D\tilde{C}_\varepsilon$ being the square matrix in \eqref{Ce} and $V$ denoting the $(N-k+1)$-dimensional Lebesgue measure. Subtracting in $D\tilde{C}_\varepsilon$ the first line to all the other lines, we obtain that

\[
 \det D\tilde{C}_\varepsilon=\det
\begin{pmatrix}
1-\varepsilon +\frac{k\varepsilon}N & \frac{\varepsilon}N &\frac{\varepsilon}N &  \dots&\frac{\varepsilon}N  \\
\varepsilon-1& 1-\varepsilon& 0 &  \dots&0 \\
\varepsilon-1 & 0 & 1-\varepsilon &  &0 \\
\varepsilon-1 & 0  & 0 & \ddots & \vdots\\
\vdots &\vdots & \vdots&  &1-\varepsilon\\
\end{pmatrix}.
\]
A Laplace expansion along the first line gives 
\begin{equation}\label{cc}
    \det D\tilde{C}_\varepsilon = (1-\varepsilon)^{N-k}.
\end{equation}

Let us denote $\mathcal{I}$ the set of domains of injectivity of $\tilde{F}_0$. Then, we have that

\begin{equation}\label{tt}
\begin{aligned}
    V(\tilde{F}_0(\mathcal{S})) &\le \sum_{I\in \mathcal{I}}\int_{I\cap\mathcal{S}} |\det D\tilde{F}_0(\tilde{x})| dV(\tilde{x})\\
                       & \le \int_\mathcal{S} |\det D\tilde{F}_0(\tilde{x})| dV(\tilde{x})\\
                       & = \int_\mathcal{S} \overset{N-k+1}{\underset{i=1}\prod} |f'(\tilde{x}_i)| dV(\tilde{x})\\
                       &\le \underset{\tilde{x}\in \mathcal{S}}\sup \overset{N-k+1}{\underset{i=1}\prod} |f'(\tilde{x}_i)| V(\mathcal{S} ).
\end{aligned}
\end{equation}
From \eqref{==}, \eqref{ss} and \eqref{cc} it follows that 
\begin{equation}
 V(\mathcal{S})\le(1-\varepsilon)^{N-k} \underset{\tilde{x}\in \mathcal{S}}\sup \overset{N-k+1}{\underset{i=1}\prod} |f'(\tilde{x}_i)| V(\mathcal{S} ).
\end{equation}
As by hypothesis
\begin{equation}\label{u}
   (1-\varepsilon)^{k-N} \ge \underset{\tilde{x}\in \mathcal{S}}\sup \overset{N-k+1}{\underset{i=1}\prod} |f'(\tilde{x}_i)|,
\end{equation}
we must have $V(\mathcal{S} )=0$.

In order to show that the non-wandering set $\Omega$ of $\tilde{F}_\varepsilon$ satisfies $V(\Omega)$=0, we prove that $\Omega\subset \mathcal{S}$. To proceed, we recall that $\Omega$  is the set of the points $x\in\tilde{X}$ such that for any neighborhood  $U$ of  $x$ and any  $n_0\ge1$, there exists $n\ge n_0$ such that $\tilde{F}_\varepsilon^{-n}(U)\cap U\neq \emptyset$.

Let us suppose $x\notin\mathcal{S}$ and let us show  that $x\notin\Omega$. If $x\notin \mathcal{S}$, then there exists $n_0\ge 1$ such that $x\notin \tilde{F}_\varepsilon^{n_0}(\tilde{X})$. Since $\tilde{F}_\varepsilon^{n_0}(\tilde{X})$ is compact, its complement in $\tilde{X}$ is an open set of $\tilde{X}$. Therefore, there exists a neighborhood  $U\subset\tilde{X}$ of  $x$ such that
 $\tilde{F}_\varepsilon^{n_0}(\tilde{X})\cap U=\emptyset$. This implies that  $\tilde{F}_\varepsilon^n(U)\cap U= \emptyset$ for all $n\ge n_0$, since $\tilde{F}_\varepsilon^n(U)\subset \tilde{F}_\varepsilon^{n_0}(\tilde{X})$ whenever $n\ge n_0$. It follows that  $x\notin \Omega$.
\end{proof}

\begin{remark}

1) Since the non-wandering set contains the $\omega$-limit set of any points of $\tilde{X}$, under condition \eqref{cond}, any attractor of the chimeric dynamics has to be of zero Lebesgue measure. It could be a cluster of larger size (such as the diagonal), a periodic cycle, or even a non-trivial set of zero Lebesgue measure, such as the hyperbolic attractor described in the previous section.

2) As $\mu(\Omega)=1$ for any invariant probability measure $\mu$ of $\tilde{F}_\varepsilon$, the previous proposition implies that  $\tilde{F}_\varepsilon$ admits no absolutely continuous invariant measure, if \eqref{cond} holds.

3) Under condition \eqref{cond}, we have $\sup|f'|<\frac1{1-\varepsilon}$ for all $2\le k\le N$. So by point 3) of Remark \ref{erg}, $\Lambda_k^{\bot,+}(x)<0$ for all $k$ and all $x\in \mathcal{C}_k$. In particular, the diagonal can attract a set of positive Lebesgue measure. However, this set is not necessarily of full measure and there can exist other attractors in $\tilde{X}$. Proposition \ref{the2} ensures that they are of zero Lebesgue measure if \eqref{cond} holds.

\end{remark}

\section{Conclusion and comments}
We were able to compute the TLE for globally coupled map lattices by exploiting the important symmetries of this class of systems. It is not clear whether one could provide analogous results for different kinds of coupling operators, such as the ones considered in \cite{batista}. In particular, the cluster sets $\mathcal{C}_J$ are not always invariant for other types of couplings. Overall, our results provide a better understanding of the conditions of formation and the long-term behavior of chimeras. They also confirm the existence of attracting chimera states evolving on a periodic cycle and suggest the existence of chaotic chimeras attracting a massive set of initial conditions. Our formula for the TLE provides a valuable tool to study the stability of cluster spaces from a numerical perspective. This allows for instance to discriminate between dummy chimeras caused by numerical artifacts and genuine attracting chimeras. Our study suggests that attracting chimeras whose dynamics settles on a set of positive Lebesgue measure in chimera spaces are implausible, although we did not prove them impossible for this type of coupling. We hope that our results will motivate future works that will either construct such an example or rule out their existence.

\section{Acknowledgements}
TC was partially supported by CMUP, which is financed by national funds through FCT – Fundaç$\tilde{a}$o para a
Ci\^encia e Tecnologia, I.P., under the project with reference UIDB/00144/2020. The authors thank Sandro Vaienti for his precious comments, which contributed to improving the quality of this paper. TC thanks Dylan Bansard-Tresse whose help was determinant in some of the computations of section 5.\\


\begin{thebibliography}{50}
\bibitem{exp1} M. R. Tinsley, S. Nkomo, and K. Showalter, Chimera and phase-cluster states in populations of coupled chemical oscillators. Nature Physics, 8(9):662–
665 (2012).

\bibitem{exp2} E. A. Martens, S. Thutupalli, A. Fourriere, and O. Hallatschek, Chimera states in mechanical oscillator networks. Proceedings of the National Academy of Sciences, 110(26):10563–10567 (2013).

\bibitem{ke} G. Keller, Exponents, attractors and Hopf decompositions for interval maps, Ergodic
Theory Dynam. Systems, 10 717–744 (1990).

\bibitem{bezout} {\em https://terrytao.wordpress.com/2011/03/23/bezouts-inequality/}

\bibitem{fewn} D. J. Bates, F. Bihan, F. Sottile, Bounds on the Number of Real Solutions to Polynomial Equations, Int. Math. Res. Not. IMRN, 23 7 (2007).

\bibitem{stro} D. M. Abrams, S. H. Strogatz, Chimera states for coupled oscillators, Phys. Rev. Lett. 93(17):174102 (2004).

\bibitem{bick} C. Bick, P. Ashwin, Chaotic weak chimeras and their persistence in coupled populations of phase oscillators. Nonlinearity,
29(5):1468 (2016).

\bibitem{alexander} J.C. Alexander, J. A. Yorke, Z. You, I. Kan, Riddled basins, International Journal of Bifurcation and Chaos 4 795-813 (1992).


\bibitem{ashwin} P. Ashwin, J. Buescu, I. Stewart, From attractor to chaotic saddle: a tale of transverse instability, Nonlinearity 9 703–737 (1996).


\bibitem{batista} C. Anteneodo, S. E.de Souza Pinto, A. M. Batista, R. Viana, Analytical results for coupled-map lattices with long-range interactions, Phys. rev. A, Atomic, molecular, and optical physics 68(4) (2003).

\bibitem{piko} Arkady Pikovsky, Antonio Politi, Lyapunov Exponents: A Tool to Explore Complex Dynamics, Cambridge U. Press (2016).

\bibitem{cosenza} M. G. Cosenza, O. Alvarez-Llamoza and A. V. Cano, Chimeras and clusters emerging from robust-chaos dynamics, preprint, arxiv.org/abs/2102.06087 (2021).





\bibitem{asta} J. Stoer, R. Bulirsch, Introduction to Numerical Analysis, Springer, New York (2002).

\bibitem{walters} P. Walters, An Introduction to Ergodic Theory, Springer-Verlag, New York (1982).

\end{thebibliography}
\end{document}